\documentclass[10pt, conference]{IEEEtran}

\usepackage[utf8]{inputenc} 
\usepackage[english]{babel} 
\usepackage{mathtools}
\usepackage{amsmath}
\usepackage{enumitem}
\usepackage{amssymb}
\usepackage{amsthm}
\usepackage{booktabs}  
\usepackage{subcaption} 
\usepackage[switch, columnwise]{lineno}
\usepackage{thmtools}
\usepackage{ebproof}

\usepackage[T1]{fontenc} 
\usepackage{stmaryrd} 
\usepackage{cmll}
\usepackage[maxnames=10]{biblatex}
\usepackage{tikz}
\usepackage{tikz-cd}
\usepackage{tkz-graph}
\usetikzlibrary{matrix,arrows,decorations.pathmorphing,positioning}
\tikzset{shorten <>/.style={shorten >=#1,shorten <=#1}}
\tikzcdset{scale cd/.style={every label/.append style={scale=#1},
    cells={nodes={scale=#1}}}}
\newtheorem{theorem}{Theorem}[section]
\newtheorem{lemma}[theorem]{Lemma}
\newtheorem{proposition}[theorem]{Proposition}
\newtheorem{corollary}[theorem]{Corollary}
\newtheorem{example}[theorem]{Example}

\declaretheoremstyle[    headfont=\bfseries, bodyfont=\upshape]{normalhead}
\declaretheorem[style=normalhead]{definition}

\newcommand{\Set}{\mathrm{Set}}
\newcommand{\Cat}{\mathrm{Cat}}
\newcommand{\Dist}{\mathrm{Dist}}

\newcommand*{\Ob}[1]{\mathrm{ob}(#1)}

\newcommand{\CAT}{\mathrm{CAT}}

\newcommand{\MRel}{\mathrm{MRel}}
\newcommand{\SCat}{\textit{S}\text{-}\mathrm{CatSym}}

\newcommand{\SDist}{\textit{S}\text{-}\mathrm{Dist}}
\newcommand*{\psalg}[1]{S\text{-PAlg}_{\bic{#1}}}
\newcommand*{\alg}[1]{S\text{-Alg}_{\bic{#1}}}
\newcommand*{\lalg}[1]{S\text{-LAlg}_{\bic{#1}}}
\newcommand{\Of}{\mathbb{O}_f}

\newcommand*\subst[3]{ #1 \sub{#3}{#2}}

\newcommand*{\sub}[2]{\{#1/#2\}}

\newcommand*{\fv}[1]{\mathsf{fv}(#1)}

\newcommand*{\length}[1]{\mathsf{len}(#1)}

\newcommand*{\sem}[2]{ \llbracket #2 \rrbracket_{#1}}

\newcommand*{\size}[1]{\text{s}\left(#1\right)}

\newcommand*{\subis}[4]{\mathsf{sub}^{#1,#2,#3, #4}}

\newcommand*{\seq}[1]{\langle #1 \rangle}

\newcommand*{\seqdots}[3]{\seq{#1_{#2}, \dots, #1_{#3} } }

\newcommand*\rappl[2]{#1 #2}
\newcommand*\appl[2]{#1 #2}
\newcommand*{\la}[1]{\lambda #1.}

\newcommand*{\todi}{\nrightarrow}

\newcommand*{\Tm}{M}
\newcommand*{\Tmtwo}{N}

\newcommand*{\morpCone}{\eta}

\newcommand*{\ty}{a}

\newcommand*{\tyl}{\vec{a}}

\newcommand*{\bic}[1]{\mathcal{#1}}

\makeatletter
\newcommand*{\doublerightarrow}[2]{\mathrel{
  \settowidth{\@tempdima}{$\scriptstyle#1$}
  \settowidth{\@tempdimb}{$\scriptstyle#2$}
  \ifdim\@tempdimb>\@tempdima \@tempdima=\@tempdimb\fi
  \mathop{\vcenter{
    \offinterlineskip\ialign{\hbox to\dimexpr\@tempdima+1em{##}\cr
    \rightarrowfill\cr\noalign{\kern.5ex}
    \rightarrowfill\cr}}}\limits^{\!#1}_{\!#2}}}
\newcommand*{\triplerightarrow}[1]{\mathrel{
  \settowidth{\@tempdima}{$\scriptstyle#1$}
  \mathop{\vcenter{
    \offinterlineskip\ialign{\hbox to\dimexpr\@tempdima+1em{##}\cr
    \rightarrowfill\cr\noalign{\kern.5ex}
    \rightarrowfill\cr\noalign{\kern.5ex}
    \rightarrowfill\cr}}}\limits^{\!#1}}}
\newcommand{\colim@}[2]{\vtop{\m@th\ialign{##\cr
    \hfil$#1\operator@font lim$\hfil\cr
    \noalign{\nointerlineskip\kern1.5\ex@}#2\cr
    \noalign{\nointerlineskip\kern-\ex@}\cr}}}
\newcommand{\colim}{%
  \mathop{\mathpalette\colim@{\rightarrowfill@\textstyle}}\nmlimits@
}
\makeatother

\ifCLASSINFOpdf
\else
\fi
\IEEEpubid{111}

\begin{document}
%
\title{Intersection Type Distributors}
\author{\IEEEauthorblockN{Federico Olimpieri} \IEEEauthorblockA{LIPN, Université Sorbonne Paris Nord, France\\Email: olimpieri@lipn.univ-paris13.fr}}


%


\IEEEoverridecommandlockouts
\IEEEpubid{\makebox[\columnwidth]{978-1-6654-4895-6/21/\$31.00~
\copyright2021 IEEE \hfill} \hspace{\columnsep}\makebox[\columnwidth]{ }}
\maketitle

\begin{abstract}
We study a family of distributors-induced bicategorical models of $\lambda $-calculus, proving that they can be syntactically presented via intersection type systems. We first introduce a class of 2-monads whose algebras are monoidal categories modelling resource management. We lift these monads to distributors and define a parametric Kleisli bicategory, giving a sufficient condition for its cartesian closure. In this framework we define a proof-relevant semantics: the interpretation of a term associates to it the set of its typing derivations in appropriate systems. We prove that our model characterize solvability, adapting reducibility techniques to our setting. We conclude by describing two examples of our construction.\end{abstract}

%

\section{Introduction}

\subsubsection{A Logical Approach to Resources}

The notion of resource is very important in Computer Science. A resource can be copied or deleted, and these two basic operations affect the behavior of programs. Hence, a mathematical approach to the notion of resource is naturally required, as it can clarify the understanding of how programs behave. A well-known resource-sensitive mathematical framework is \emph{linear logic}, introduced by Girard \cite{gir:ll} in the $ 80s .  $ The decomposition of the intuitionistic arrow
\[ A \Rightarrow B = \oc A \multimap B \]
expresses the general non-linear behaviour of programs. The $ \oc $ construction says that we are allowed to copy or delete the input as many times as needed. Linear logic is thus immediately connected to \emph{quantitative} aspects of computation.

\subsubsection{Resources \textit{via}  Types}

A few years before Girard's introduction of linear logic, Coppo and Dezani \cite{dez:int} proposed \emph{intersection types}, a type-theoretic framework sensitive to the fact that a $ \lambda $-term can be typed in several ways. In order to define an intersection type system, they add another constructor to the syntax: $ a \cap b $. Then typability with an intersection type is equivalent to being typable with \emph{both} types $a $ and $b .$ This kind of type disciplines proved to be very useful to characterize fundamental notions of normalization in $ \lambda $-calculus ($ e.g.,$ head-normalization, $\beta $\nobreakdash-normalization, strong normalization)  \cite{kri:lam, bern:strong, kes:int}. Moreover, if the intersection type $a \cap b $ is  \emph{ non-idempotent} \cite{gard:int, carv:sem}, $ i.e. ,$ $a \cap a \neq a ,$ the considered type system is \emph{resource sensitive}. In that case, the arrow type 
\[ a_1 \cap \dots \cap a_k \Rightarrow a     \] 
encodes the \emph{exact} number of times that the program needs its input during computation. The resource awareness of non-idempotent intersection has been used to prove normalization and standardization results by \emph{combinatorial} means \cite{kes:int}, to study infinitary computation \cite{vial:klop} and to express the execution time of programs and proof-nets \cite{carv:sem, carv:weak, carv:strong}. The non-idempotent intersection type system $ \mathcal{R}, $ is also strongly connected to the \emph{Taylor expansion} of $ \lambda $-terms \cite{er:tay, carv:sem}. Thus, resource sensitive intersection corresponds also to \emph{linear approximation}. Another important feature of intersection type systems is that they determine a class of \emph{filter models} for pure $ \lambda $-calculus \cite{dez:filt}. The correspondence between intersection types and Engeler-like models is also well-known \cite{hyland:engeler}. Hence intersection types are both  \emph{syntactic} and \emph{semantic} objects. 

\subsubsection{A Categorical Approach}

The semantic side of intersection types is connected also to categorical semantics. A simple and informative categorical model for $ \lambda $-calculus is the \emph{relational model} ($ \MRel$)\footnote{For a general survey on relational semantics we refer to \cite{ong:rel}. See also \cite{manzo:rel} for results on the lambda-theories induced by this kind of models.}. Objects of $\MRel $ are sets, while morphisms are \emph{multirelations} $ f \subseteq M_{f}(A) \times B $, where $ M_f(A) $ is the free commutative monoid over $A $.  This model arises from the linear logic decomposition. It is well-known that this relational semantics  corresponds to the non-idempotent intersection type system $ \mathcal{R} $ \cite{carv:sem}. This correspondence says that the categorical interpretation of a $\lambda$-term can be presented in a \emph{concrete} way, as a form of typing assignment. In particular, the intersection type constructor $ \cap $ corresponds to the product in the free commutative monoid construction that gives the interpretation of the linear logic exponential connective. This fact suggests the possibility to model, in all generality, intersection types \textit{via} \emph{monads}. With some relevant modifications, one can also achieve in this way an idempotent intersection \cite{er:collapse, er:cbpv}. 

\subsubsection{Lifting to Bicategories}

The idea of a bidimensional semantics for $ \lambda $-calculus was first presented by Seely \cite{Seely:2sem} and further studied in \cite{tom:2cat}. The passage from $ 2$-category to bicategories consists in a weakening of the structure. In particular, associativity and identity laws for horizontal composition are now only up to coherent isomorphisms. In this setting, there is a natural generalization of the category of relations: the bicategory of \emph{distributors} ($ \Dist $). A relation $ f \subseteq A \times B $ is the same as its characteristic function $   \chi_f : A \times B \to \{ 0, 1 \} . $ In particular, the former function naturally induces a functor from $A \times B $, taken as discrete category, to the 2 elements category. It is then natural to relax the hypothesis and consider functors of the shape $     F : B^o \times A \to \Set $ where $ A $ and $ B $ are arbitrary small categories. These functors are called distributors\footnote{Another popular name for this kind of structures is \emph{profunctor}. }. Cattani and Winskel \cite{winsk:prof} proposed a distributor-induced semantics of concurrency. In particular, they also gave a distributor model of linear logic, generalizing Scott's domains. In subsequent papers, Fiore, Gambino, Hyland and Winskel \cite{fiore:mods, fiore:esp} introduced the bicategory of \emph{generalized species of structures $(\mathsf{Esp})$}, a rich framework encompassing both multirelations and Joyal's combinatorial species \cite{joyal:esp}. They also proved that $\mathsf{Esp}$ is cartesian closed and, hence, a bicategorical model for $ \lambda $-calculus. 
 
Inspired by their result, Tsukada, Asada and Ong \cite{tao:gen, tao:pdist} showed that the generalized species semantics of $\lambda $-calculus has a syntactic counterpart in the \emph{rigid Taylor expansion} of $ \lambda $-terms. At the same time, building on \cite{mel:ts, hyland:lam}, Mazza, Pellissier and Vial \cite{mazza:pol} presented a higher categorical approach to intersection types and linear approximation, rooted in the framework of \emph{multicategories} and discrete distributors. 

This bicategorical setting has several advantages. First, one can model term rewriting in a categorical way, \textit{via} 2-cells between terms denotations. Second, as shown in \cite{mazza:pol}, 2-dimensional categorical constructions can be useful to reason \emph{parametrically} on syntax, sparing a lot of time that normally is lost in checking special cases. Third,  distributors are an example of \emph{categorification}: the set-theoretic notion of relations is replaced by a category-theoretic one. It turns out that this replacement makes explicit relevant information that was hidden in the non-categorified setting.

\subsubsection{Our Contribution}

Building on \cite{fiore:esp, fiore:rel, joyal:op, tao:gen, marsd:quant}, we define a family of distributor-induced denotational semantics of $\lambda$-calculus. These bicategories of distributors are Kleisli bicategories for an appropriate collection of pseudomonads, the \emph{resource monads}. These are 2-monads over categories, whose algebras are some special kind of strict monoidal categories. We can sum up the results of the paper in a procedural way: 
\begin{enumerate}
\item[(i)] Take a resource monad $S$ and apply the construction of \cite{fiore:rel} to obtain a pseudomonad $\tilde{S}$ (Section \ref{resourcemon}).
\item[(ii)] Consider the Kleisli bicategory of $ \tilde{S}$, $ \SDist .$ Its opposite bicategory $(\SDist )^o = \SCat$, the bicategory of $S $-\emph{symmetric sequences},\footnote{ A parametric generalisation of the bicategory of categorical symmetric sequences introduced in \cite{joyal:op} and biequivalent to generalized species.} is cartesian closed if the algebras of $S $ are symmetric strict monoidal categories (Section \ref{scat}). 
\item[(iii)] Consider the $ \lambda $-calculus semantics induced by $ \SCat$ (Section \ref{sem}). Following the construction presented in Section \ref{denint}, get the parametric \emph{category of types} $D_A $ and intersection type system $ E_A,$ generated by a small category $A $ of \emph{atomic types} and the resource monad $S $. 
\item[(iv)] By the results of Section \ref{intred}, the considered type system is a \emph{proof relevant} denotational semantics for $ \lambda$\nobreakdash-calculus. The distributor that interprets a $ \lambda $-term $ M $, its \emph{intersection type distributor}, is defined in the following way:
\[  \sem{\vec{x}}{M}(\Delta, a) = \left \{  \begin{prooftree}
\hypo{ \tilde{ \pi }  }
\ellipsis{}{ \Delta &\vdash M : a} \end{prooftree} \right\}  \] 
where $ \tilde{\pi} $ is an equivalence class of \emph{typing derivations}, $a $ is a type and $ \Delta $ is a type context.     
 \end{enumerate}
 The equivalence relation is induced by composition in the appropriate bicategory $\SCat. $ The equivalence is crucial, since it forces the preservation under reduction not only of \emph{typability}, but of the \emph{amount} of classes of typing derivations. We remark that this refines and improves the standard relational semantics, where the denotation of a term is just a \emph{test of typability} and do not give any information about derivations.
As in \cite{mazza:pol}, our construction gives rise to four intersection type systems, \emph{linear}, \emph{affine}, \emph{relevant} and \emph{cartesian} ones.
The structure of the resource monad $S $ gives the kind of intersection connective. For example, the $2 $\nobreakdash-monad for symmetric strict monoidal categories determines a non-idempotent (linear) intersection. By contrast, the $ 2$-monad for cartesian categories determines an intersection that admits \emph{duplication} and \emph{erasing} of resources. 

Moreover, our model internalizes subtyping in a categorical framework: the preorder relation $ a \leq b$ between intersection types is replaced by an arrow $ f : a \to b $ in an appropriate \emph{category of types}. The intuition behind it is that $f $ is now a \emph{witness} of subtyping.  The construction of morphisms between types naturally generalizes the standard subtyping rules, as expected. 

The strength of our approach is twofold. First, we are able to give a concrete presentation of a relevant class of quite abstract and esoteric semantics for $  \lambda $-calculus. Second, this presentation determines a parametric theory of intersection types. In particular, our theory can account for proof-relevance, subtyping and denotational semantics.

\subsubsection{Discussion of Related Work}

\begin{itemize}

\item[(i)] Our approach is independent from \cite{fiore:bicatt, fiore:ttbic2, saville:cc}. Fiore and Saville presented a bicategorical extension of simply typed $\lambda $-calculus that corresponds to the appropriate type theory for cartesian closed bicategories. Our intersection type systems can be seen as an approximation theory for simple types and arise by making explicit the structure of a \emph{special class} of bicategorical models.

\item[(ii)]We vastly generalize the results of \cite{ol:bang}, where  a categorification of non-idempotent intersection type is presented. Our parametric construction over resource monads determines a categorification of linear (non-idempotent), affine, relevant and cartesian (idempotent) ones. Then \cite{ol:bang} becomes a special case of our method (Section \ref{app})\footnote{The only sensible difference is that while in \cite{ol:bang} an untyped call-by-push value calculus \cite{levy:cbpv, gg:bang, guerr:bang} is considered, in the present setting we chose pure $ \lambda $-calculus. Our choice is only instrumental to avoid additional technicalities. }.

 \item[(iii)]  In \cite{mazza:pol} a parametric 2-categorical construction of intersection type systems is presented. Intersection type systems are seen as special kind of \emph{fibrations}. This contribution can be seen as a ``syntactic categorification'' of intersection types. Indeed, while the construction of Mazza \textit{et al.} is an elegant and very general approach to intersection type disciplines, that also allows to prove normalization theorems in a modular way, it does not provide a \emph{type-theoretic denotational semantics}\footnote{Given $ M \to_{\beta} N , $ the type theoretic structure associated to $ M $ is not, in general, isomorphic to the one of $ N $ \cite{mazza:hdr}[pp. 65-66]. }. Moreover, their work is limited to the \emph{discrete} case, \textit{i.e.}, they do not consider subtyping. We shall see that what is needed to obtain both denotational semantics and subtyping is highly non-trivial.
 
 \item[(iv)] Our work is closely related to the rigid Taylor expansion semantics \cite{tao:gen}. However, Tsukada \emph{et al.} contribution is restricted syntactically to $\eta $-long simply typed terms and semantically to generalized species over groupoids. The generalization of their approach to the whole simply typed and untyped $ \lambda $-calculus and to the parametric bicategory $ \SCat $ is, again, highly non-trivial and is one of the goals of our work\footnote{However, in the present paper we consider intersection types instead of terms approximations. In \cite[Chapter 4]{ol:thesis}  is shown that a naive generalization of Tsukada, Asada and Ong's approach fails. The general notion of approximation needs to take into account the subtyping information given by typing derivations. }.
\end{itemize}

\subsubsection{Outline}
Section \ref{back} introduces some categorical background. The main goal is to define a family of Kleisli bicategories of distributors, associated with the lifting of a suitable collection of doctrines.
In Section \ref{sem} we build a family of bicategorical non-extensional models for $\lambda $-calculus and we use it to define an intersection type-theoretic denotation for terms. In Section \ref{h-norma} we give a parametric proof of the head-normalization theorem for our denotational models. We conclude in Section \ref{app} by considering two concrete examples of our construction.
\subsubsection{Notations}
  Given a category $C $ we write $C^{o} $ for its opposite category. Given a bicategory $\bic{C} $ we write $ \bic{C}^{op} $ for the bicategory obtained by reversing the $ 1$\nobreakdash-cells of $\bic{C} $ but not the $2 $-cells. We write $\CAT $ for the $2 $\nobreakdash-category of locally small categories, functors and natural transformations and $\Cat $ for its full sub-2-category of small categories. Given (bi)categories $ A_1, \dots, A_n $ we denote as $ \prod_{i=1}^{n} A_i $ their product. Given (bi)categories $ A_1, \dots, A_n $ we denote as either $ \bigsqcup_{i=1}^{n} A_i $ or $ \sum_{i=1}^{n} A_i$ their coproduct. Given categories $ A,B$, we use either $ [A,B] $ or $ \CAT(A,B) $ to denote their functor category. We denote the initial category as $ \emptyset .$ We use linear logic notations for the general notions of cartesian product, terminal object, etc.

\section{Categorical Background} \label{back}

We suppose that the reader is familiar with the basics of bicategory theory, for which we refer to \cite{borc:cat}.

\subsection{Integers and Lists}

 We consider the category $ \Of $ where object are \emph{finite ordinals} $ [n] = \{ 1, \dots, n \} $, for $ n \in \mathbb{N}, $ and morphisms are functions.  The category $ \Of $ is symmetric strict monoidal, with tensor product given by addition: $[n] \oplus [m] = [n + m] .$ Let $ k_1, \dots, k_n  $ be natural numbers and  $ \alpha : [m] \to [n] $ we define $ \bar{\alpha} : [\sum_{ j = 1   }^{m} k_{\alpha(j)}] \to [ \sum_{ i = 1  }^{n} k_{i} ]  $ as follows: \[ \bar{\alpha}(\sum_{j = 1}^{l-1} k_{\alpha(j)} + p  ) = \sum_{ i = 1}^{\alpha(l)-1} k_i + p  \]with $  l \in [m] ,  $ and $  1 \leq p \leq k_{\alpha(l)} . $   If we apply the former construction to bijections, we get the symmetries of the tensor product. 
 
  From $ \Of $ we can build categories of indexed families of objects over finite ordinals. Let $ \seq{a_1, \dots, a_k} $ be a list of elements of $A .$  We write $ \length{\vec{a}} $ for its length. We denote lists as $ \vec{a}, \vec{b}, \vec{c} \dots $ Given a list $ \vec{a} = \seqdots{\ty}{1}{k} $ and a function $ \alpha : [k] \to [k'] $ we define the \emph{right action of $\alpha $ on $ \vec{a} $} as $\vec{a}\{\alpha \} = \seqdots{\ty}{\alpha(1)}{\alpha(k)} .  $ Given a category $A$, we define the category $ \Of A$  of \emph{lists} of $A $, as follows: 

\begin{enumerate}
\item $ \text{Obj}(\Of A) = \{ \seq{a_1, \dots, a_n} \mid a_i \in A      \}. $
\item $ \Of A (\seq{a_1,\dots, a_n}, \seq{b_1, \dots, b_m}) = \{  \seq{\alpha, f_{1}, \dots, f_{m}} \mid \alpha : [m] \to [n] $ and $ f_{i} : a_{\alpha(i)} \to b_i       \}. $
\item For $ \seq{\alpha, \vec{f}} : \vec{a} \to \vec{b} $ and $ \seq{\beta, \vec{g}} : \vec{b} \to \vec{c}, $ composition is given by 

\[       \seq{\beta, \vec{g}} \circ  \seq{\alpha, \vec{f}} = \seq{\alpha \circ \beta, \vec{g} \circ \vec{f}\{\alpha \}  }     \]
\end{enumerate}

The category $ \Of A $ is cartesian monoidal, with products given by lists concatenation. For $ \seq{\tyl_1, \dots, \tyl_n} $ and $ \alpha : [m] \to [n] $  with $ \length{\vec{a}_i} = k_i $ we define \[ \alpha^{\star} : \bigoplus_{i = 1}^{n} \tyl_i \to \bigoplus_{j= 1}^{m} \tyl_{\alpha(j)} \] as $ \alpha^{\star} = \seq{\bar{\alpha}, \vec{1}_{\bigoplus_{i=1}^{k}\vec{a}_{\alpha(i)}}  } . $ The former construction encompasses all arrows that are compositions of \emph{structural morphisms}, \textit{i.e.}, of symmetries, terminal arrows, projections and diagonals.

\subsection{Coend calculus}

Virtually everything that follows is rooted in the notion of \emph{coend}. 

\begin{definition}
Let $ F : C^o \times C \to D $ be a functor. A \emph{cowedge} for $ F $ is an object $  T \in D$ together with a family of morphisms $ w_{c} : F(c, c) \to T $ such that the following diagram commutes 

\[\begin{tikzcd}[column sep= 2 cm, row sep= 2 cm]
 F(c',c) \arrow{r}[name=a]{F(f,1)}   \arrow{d}[name=d]{F(1,f)}   & 
F(c,c)  \arrow{d}[name=e]{w_c }    \\
 F(c',c')  \arrow{r}[name=b]{w_{c'} }                                   &  
T & 
\end{tikzcd} \]

for $ f : c \to c' . $

\end{definition}

A \emph{coend} is then an universal cowedge. We denote the coend of $ F $ as $ \int^{c \in C} F(c,c) .$ Clearly a coend is a kind of colimit, precisely a coequalizer. The integral notation is justified by the formal calculus connected with this notion\footnote{For a proper introduction to coend calculus see \cite{fosco:coend}.}. 



\subsection{Presheaves} 
 
For a small category $A $ define $ PA =[A^o, \Set]  ,$ the category of \emph{presheaves} of $A  $ and natural transformations. If $ A$ is monoidal, for $ X, Y \in PA, $ we define the \emph{Day convolution} tensor product \cite{day:con} pointwise 

\[ (X \hat{\otimes} Y)(a) = \int^{a_1, a_2 \in A} X(a_1) \times Y(a_2) \times A(a, a_1 \otimes a_2)  . \]

It is well-known and crucial that $ PA $ is the \emph{free cocompletion} of $A$. This derives directly from the \emph{Yoneda embedding} and what is called the \emph{density theorem}, \textit{i.e.,} that presheaves are canonical colimits of representables. The freeness condition is then satisfied by the \emph{left Kan extension}: 
  \[\begin{tikzcd}
A \arrow[r, "Y_A"] \arrow[d, "F"'] & PA \arrow[ld, bend left, "{L_{Y}(F)}"]  \\
B                                        &                  
\end{tikzcd}\]
Where $B $ is a cocomplete category, $ Y_A $ is the Yoneda embedding and $F $ functor.

\subsection{Distributors}

We now define the bicategory of \emph{distributors}.
\begin{itemize}
\item 0-cells are small categories $A,B,C \dots $;
\item 1 cells $ F: A\nrightarrow B $ are functors $ F: B^{o} \times A \to \Set.  $ By the cartesian closed structure of the 2-category of categories, functors and natural transformations we have the following correspondence: 

\[ \begin{prooftree}
\hypo{ F: B^{o} \times A \to \Set }
	\infer1[]{ F^{\lambda} : A \to PB }
\end{prooftree} \]
Hence we will switch from one to the other presentation of distributors when convenient. 
\item 2-cells $ \alpha : F \Rightarrow G $ are natural transformations.
\item For fixed 0-cells $A$ and $B$, 1-cells and 2-cells organize themselves as a category $ \Dist(A,B). $ Composition $ \alpha \star \beta  $ in $ \Dist(A,B) $ is called \emph{vertical composition}.
\item For $ A \in \Dist$, the identity $1_{A} : A \nrightarrow A $ is defined as the Yoneda embedding  $ 1_{A}(a, a') = A(a, a'). $ 
\item For 1-cells $ F : A \nrightarrow B $ and $ G : B \nrightarrow C $ the \emph{horizontal composition} is given by

\[ (G \circ F) (c,a) = \int^{b \in B} G(c,b) \times F(b,a).  \]

associative and identities are only up to canonical isomorphism. For this reason $ \Dist $ is a bicategory \cite{borc:cat} .   

\item There is a symmetric monoidal structure on $ \Dist $ given by the cartesian product of categories:
$ A \otimes B = A \times  B .$ The bicategory of distributors is compact closed and orthogonality is given by taking the opposite category $ A^{\bot} = A^{o} . $ The linear exponential object is then defined as $ A \multimap B = A^o \times B . $
\item For $A, B \in \Ob{\Dist} $ there is a \emph{zero distributor} $ \emptyset_{A,B} \in \Dist(A,B)  $ such that for all $ \seq{b, a} \in B \times A , \emptyset_{A,B}(b,a) = \emptyset . $ 
\end{itemize} 

Given a functor $ F : A \to B $ we can define distributors $ \bar{F} : A \todi B, \underline{F} : B \todi A $ as $ \bar{F}(b,a) = B(b, F(a)) $ and $\underline{F}  = B(F(a),b)$\footnote{The two distributors are adjoint 1-cells in the bicategory $\Dist$.}.

\subsection{Pseudomonads and Algebras}\label{resourcemon}

For a proper introduction to two-dimensional monad theory we refer to \cite{kelly:2mon}.

\begin{definition}
Let $ \bic{C}$ be a $2 $-category. A \emph{$2 $-monad} over $\bic{C}$ is a triple $(T, m, e) $ where $T $ is a $2 $-endofunctor on $\bic{C} $ and $ m : T^{2} \to T ,$ $ e : 1 \to T $ are $2 $-natural transformations satisfying the usual monadic commutative diagrams. A \emph{pseudomonad} over $ \bic{ C}$ is the same as a $2$-monad but the commutation of diagrams is only up to coherent isomorphisms.
\end{definition}

Given a 2-monad $ \seq{S : \mathcal{C} \to \mathcal{C}, \eta, \mu} $ we can build the \emph{category of lax algebras of $S$}, $ \lalg{C} $ as follows: 

\begin{itemize}
\item An object of $\lalg{C}$ is given by an object $ A \in \bic{C},$ called the \emph{underlying object},  a 1-cell $ h_A : SA \to A $  called the \emph{structure map} and 2-cells $  \iota_1, \iota_2  $:
\[\begin{tikzcd}
SSA \arrow[r, "\mu_A"] \arrow[d, "h_{SA}"'] & SA \arrow[d, "h_A"] \ar[d]\ar[ld,shorten <>=10pt,Rightarrow, "\iota_1"] \\
SA \arrow[r, "h_A"']                        & A                  
\end{tikzcd} \qquad \begin{tikzcd}
A \arrow[r, "\eta_{SA}"] \arrow[rd, equal, ""{name=a}] & SA \arrow[d, "h_{SA}"]\ar[d]  \\
                                                         & A   & 
                                                         \arrow[Rightarrow, "\iota_2", from=1-2, to=a]
                                                        \end{tikzcd}                    \]
The $ 2 $-cells need to verify 2 additional coherence conditions \cite{luc:lax}. We denote lax algebras by $ \mathbb{A}, \mathbb{B}, \dots $ If the 2-cells $ \iota_1, \iota_2 $ are isos, $ \mathbb{A} $ is called a \emph{pseudoalgebra}. If they are identities, $  \mathbb{A}$ is a \emph{strict algebra}.
 
\item For lax  algebras $ \mathbb{A}, \mathbb{B} $ a 1-cell or morphism $ \varphi : \mathbb{A} \to \mathbb{B} $ is a morphism $ F : A \to B $ together with an invertible $2 $-cell  

\[\begin{tikzcd}[column sep= 2 cm, row sep= 2 cm]
 SA \arrow{r}[name=a]{ SF}   \arrow{d}[name=d]{h_A}   & 
SB  \arrow{d}[name=e]{ h_B }  \arrow[ld, Rightarrow, shorten >= 30pt,shorten <= 30pt, "\zeta"]  \\
 A  \arrow{r}[name=b]{ F }                                   &  
B  & 
\end{tikzcd} \]

required to satisfy two coherence conditions \cite{kelly:2mon}[p.3]. If $ \zeta $ is an isomorphism, then the morphism is called a \emph{pseudomorphism}. If $ \zeta $ is the identity, then the morphism is called a \emph{strict} morphism.
\item The category of lax-algebras can be also equipped with a 2-dimensional structure \cite{kelly:2mon}. 
\end{itemize}
 We denote the 2-categories of \emph{pseudoalgebras} and \emph{strict algebras} as respectively $ \psalg{C} $ and  $ \alg{C} $, in both cases the 1-cell considered are pseudomorphisms. Clearly we have that $ \alg{C} $ is a full 2-subcategory of $ \psalg{C} . $

\subsubsection{Resource Monads}

We present a list of $2 $-monads over $\CAT, $ the $ 2$-category of locally small categories, functors and natural transformations. We follow the spirit of \cite{marsd:quant}.  We call these monads \emph{resource monads.} The intuition is that each of these monadic constructions gives a particular notion of resource management.

We start by giving a canonical presentation of some free monoidal constructions. We assume that the reader is familiar with monoidal categories. We explicitly denote an arbitary (symmetric) monoidal category as $\mathbb{A} = \seq{C, \otimes, 1, \alpha, \lambda, \rho, \sigma}$, where $C $ is its underlying category, $\otimes $ its tensor, $ 1$ its unit, $\alpha $ is the associator, $ \lambda, \rho $ the unitors and $ \sigma$ the symmetry.

\begin{definition}
A \emph{semicartesian monoidal category} is a  symmetric monoidal category $ \mathbb{A} = \seq{C, \otimes, 1, \alpha, \lambda, \rho, \sigma}$ such that the unit is a terminal object. We write then $e_a : a \to 1 $ for the terminal morphism.
\end{definition}

\begin{definition}
A \emph{relevant monoidal category} is a symmetric monoidal category $ \mathbb{A} = \seq{C, \otimes, 1, \alpha, \lambda, \rho, \sigma}$ equipped with a natural transformation $  c_a : a \to a \otimes a   $, called the \emph{diagonal}, which has to satisfy additional coherence conditions.
\end{definition}

A monoidal category that is both semicartesian and relevant is a cartesian category.

\begin{proposition}\label{charRE}
For $A  \in \Cat$ and $ \vec{a}, \vec{b} \in \Of A $ with $ n = l(\vec{a}), m = l(\vec{b}) $ we define  

\[ \Of A^{\ast}(\vec{a}, \vec{b}) =  \sum\limits_{\alpha : [m] \to [n]} \prod\limits_{i \in [m]} A(a_{\alpha(i)}, b_i) \]

for $ \alpha : [m] \to [n] $ being restricted either to general functions, bijections, surjections, injections or identities. The following holds:
\begin{enumerate}
\item If $ \alpha $ is restricted to identities, then $ \Of A^{\ast}(\vec{a}, \vec{b})$ is the homset of the free strict monoidal category on $A. $
\item If $ \alpha $ is restricted to bijections, then $ \Of A^{\ast}(\vec{a}, \vec{b})$ is the homset of the free symmetric strict monoidal category on $A . $
\item If $ \alpha $ is restricted to injections, then $ \Of A^{\ast}(\vec{a}, \vec{b})$ is homset of the free semicartesian strict monoidal category on $A$.
\item  If $ \alpha $ is restricted to surjections, then $ \Of A^{\ast}(\vec{a}, \vec{b})$ is the homset of  free relevant strict monoidal category  on $A$.
\item If $ \alpha $ is a general function then $\Of A^{\ast}(\vec{a}, \vec{b})$ is the homset of the free cartesian monoidal strict category on $A$.
\end{enumerate}
\end{proposition}
\begin{proof}
The  proof exploits the fact that each $\Of A^{\ast}(\vec{a}, \vec{b})$ defines a subcategory of $ \Of A. $ The unit $ \eta_{A} : A \to \Of A ^{\ast}$ is given by the singleton embedding $  a \mapsto \seq{a}  . $
\end{proof}

The resource monads are then the following 2-monads. 

\begin{enumerate}
\item The \emph{strict monoidal resource monad}: the $2 $\nobreakdash-monad over $\CAT $ that sends a category $A $ to its free strict monoidal completion;
\item The \emph{linear resource monad}: the $2 $-monad over $\CAT $ that sends a category $A $ to its free symmetric strict monoidal completion;
\item The \emph{semicartesian resource  monad}: the $2 $-monad over $\CAT $ that sends a category $A $ to the free semicartesian strict monoidal category on $A $;
\item The \emph{relevant resource monad}: the $2 $-monad over $\CAT $ that sends a category $A $ to the free relevant strict monoidal category on $A $; 
\item The \emph{ cartesian resource monad}:  the $2 $-monad over $\CAT $ that sends a category $A $ to its free cartesian strict monoidal completion. 

\end{enumerate}

For $ S$ resource monad, we call the \emph{tensor product of $S $ }the tensor product on $SA $. We call \emph{$ S$-monoidal functor} a functor that preserves the structure on the nose. We denote as $ \Of^{S}  $ the full subcategory of $ \Of $ where morphisms depend on the structure of $S $ (\textit{via} Proposition \ref{charRE}).

\begin{proposition}
Let $ A, B \in \Cat $ and $S $ be a resource monad. If the tensor product of $ S $ is symmetric, then we have $ S(A \sqcup B) \simeq SA \times SB    .$
\end{proposition}

We can extend the former proposition to finite products and coproducts of categories
$ S(A_1 \sqcup \dots \sqcup A_n) \simeq SA_1 \times \dots \times SA_n $. We denote the two components of the former equivalence as respectively $ \mu_0 : S(A_1 \sqcup \dots \sqcup A_n) \to SA_1 \times \dots \times SA_n $ and
$ \mu_1 :  SA_1 \times \dots \times SA_n \to   S(A_1 \sqcup \dots \sqcup A_n) .  $

\subsubsection{The $ 2$-monadic Lifting}

In \cite{fiore:rel}, a method to extend $2 $-monads over  $\Cat $ to pseudomonads over $ \Dist $ is introduced. The construction is based on the intuition that the bicategory of distributors is the Kleisli bicategory for a suitable pseudomonad of presheaf on the $2 $-category $\Cat $. Indeed, this idea is very natural: a distributor is just a functor $ F : A \to PB .$ However, this is not strictly possible, since for a small category $A,  $ $PA $ is not small any more. In \cite{fiore:rel} the notion of relative pseudomonad is defined, in order to deal with this problem.

\begin{definition}[Relative pseudomonad]
Let $ J : \bic{C} \to \bic{D} $ be a pseudofunctor between $2 $-categories. A \emph{relative pseudomonad} $ T$ over $ J$ is the collection of the following data: 

\begin{itemize}
\item for $ A \in \bic{C} ,$ an object $TA \in \bic{D}; $
\item for $A, B \in \bic{C}, $ a functor $ (-)^{\ast}_{A,B} : \bic{D}(JA, TB) \to \bic{D}(TA,TB); $
\item for $ A \in \bic{C},$ a morphism $ i_A : JA \to TA; $
\item for $ f : A \to B $ and $g : B \to C $ a family of invertible two-cells $ \mu_{f,g} : (g^{\ast} \circ f )^{\ast}  \cong g^{\ast} \circ f^{\ast}; $
\item for $ f : JA \to TB  $ a family of invertible two cells $ \eta_f : f \cong f^{\ast} \circ i_x ; $
\item a family of invertible two cells $ i_{A}^{\ast} \cong 1_{TA} . $
\end{itemize}
This data has also to satisfy two coherence conditions \cite{fiore:rel}. 
\end{definition}

Relative pseudomonads are equipped with an appropriate notion of Kleisli bicategory \cite[Theorem 4.1]{fiore:rel}. 

Given a relative pseudomonad $ T $ over $ J :  \bic{C} \to \bic{D}  $ and a $2 $-monad $ S$ over $ \bic{C}, $ we can define a notion of \emph{lifting} of $ T$ to pseudoalgebras of $S $ \cite[Definition 6.2]{fiore:rel}. The idea is that the lifting, denoted $ \tilde{T}, $ determines a relative pseudomonad over the lifted pseudofunctor $ \tilde{J} : S $-$Alg_{\bic{C}} \to S$-$PsAlg_{\bic{D}}\footnote{In order to obtain this lifting one has to add the condition that the $2$-monad $ S$ restricts along $ J$ \cite{fiore:rel}.} . $ 

\begin{lemma}[{\cite[Example 4.2]{fiore:rel}}]
Distributors are the Kleisli bicategory for the relative pseudomonad of presheaves $ P $ over the inclusion functor $ j : \Cat \to \CAT . $
\end{lemma}

\begin{proposition}[{\cite[Theorem 6.3]{fiore:rel}}]\label{LpsA}
Let $S $ be a 2-monad over $ \bic{C} . $  If a relative pseudomonad $ T $ over $ J : \bic{C} \to \bic{D} $ lifts to pseudoalgebras of $S $, then $S $ can be extended to a pseudomonad on $ Kl(T). $ 
\end{proposition}

In our case, $  T$ will be the relative pseudomonads of presheaves $ P,$ $ Kl(P) $ the bicategory of distributors and $S $ an arbitrary resource monad.

\begin{theorem}\label{Lresmon}
Let $ S $ be a resource monad. The relative pseudomonad $ P $ lifts to the pseudoalgebras of $S .$
\end{theorem}
\begin{proof} 
For the monoidal strict monad, the symmetric monoidal strict monad and the semi-cartesian strict monad, the result was already proved in \cite{fiore:rel}. The cartesian resource monad is actually a direct corollary of the lifting of the $2 $-monad for finite products, again proved in \cite{fiore:rel}. All cases derives from a straightforward adaptation of Kelly's result on the universal property of the Day convolution \cite{kelly:day}. We prove the result for the relevant resource monad and the cartesian resource monad. In order to prove the case of symmetric strict monoidal categories with diagonals, we need to check three conditions: 
\begin{enumerate}
\item The considered monoidal structure lifts to presheaves. 
\item The Yoneda embedding preserves the considered structure. 
\item Let $ \mathbb{A} $ be $ S$-monoidal category, $ \mathbb{B} $ be a $S$-monoidally cocomplete category and $F : \mathbb{A} \to \mathbb{B} $ be a strong monoidal $S $-functor. Cocontinous functors $ F : P \mathbb{A} \to \mathbb{B} $ preserve the relevant structure.
\end{enumerate}
The three condition are verified exploiting the fact that a presheaf is a canonical colimit of representables.   
\end{proof}

\subsection{The Bicategory $\SCat $}\label{scat}
\begin{figure*}[t] 
\begin{subfigure}{9.5cm}
\centering
Types:
\[   \ty :=  o \in \Ob{A} \mid \seq{\ty_1, \dots, \ty_k} \Rightarrow \ty  \]

 Morphisms: 
\begin{gather*}
  \begin{prooftree}
\hypo{f  \in A(o,o')}
\infer1[]{ f : o \to o'}
\end{prooftree} \qquad \begin{prooftree}
\hypo{ \seq{\alpha, \vec{f}} : \tyl'\to \tyl}
\hypo{ f : \ty \to \ty'}
\infer2[]{ \seq{\alpha, \vec{f}} \Rightarrow f : (\tyl \Rightarrow \ty) \to (\tyl' \Rightarrow \ty')} \end{prooftree}     \\
\begin{prooftree}
	\hypo{\alpha \in \Of^S ([k], [k']) \quad f_{1} :  \ty_{\alpha(1)}\to \ty'_{1}\quad\cdots\quad f_{k} :  \ty_{\alpha(k)}\to \ty'_{k}}
	\infer1[]{ \seq{\alpha, f_{1},\dotsc, f_{k}} : \seq{\ty_{1},\dotsc, \ty_{k'}} \to \seq{\ty'_{1},\dotsc, \ty'_{k}}   }
\end{prooftree}
\end{gather*} 
 \caption{Category of Types $D_A $.} \label{IntD1} 
\end{subfigure} \begin{subfigure}{10cm}
\centering
Derivations:
\vspace{0.3cm}
\begin{gather*}
\begin{prooftree}
\hypo{f_1 : \vec{a}_1 \to \seq{}, \dots,  f : \tyl_i \to \seq{\ty}, \dots, f_{n} : \tyl_{n} \to \seq{} }
	\infer1[]{ x_1 :\tyl_1, \dots, x_i : \tyl_i, \dots x_n : \tyl_n \vdash x_i : \ty }
\end{prooftree}  \\  \begin{prooftree}
\hypo{ \Delta, x: \tyl \vdash M : \ty}
	\infer1[]{ \Delta \vdash \lambda x. M : \tyl \Rightarrow \ty }
\end{prooftree} \\
\begin{prooftree}
\hypo{ \Gamma_{0} \vdash M : \seq{a_1, \dots, a_k} \Rightarrow \ty }\hypo{(\Gamma_{i} \vdash N : \ty_i)_{i=1}^{k}  }
\hypo{ \morpCone :  \Delta \to  \bigotimes_{i=0}^{k} \Gamma_i }
	\infer3[]{ \Delta \vdash \rappl{M}{N} : \ty }
\end{prooftree}
\end{gather*}
\vspace{0.3cm}
\caption{Parametric Intersection Type System $E^{S}_A .$} \label{IntD2}
\end{subfigure}
\caption{Type Theoretic Presentation of the Semantics.}
\hrulefill
\end{figure*}

   \begin{figure*}[t]
\begin{gather*}
\qquad \llbracket x \rrbracket_{\vec{x}}(\Delta, a) = SD^{n}(\Delta, \seq{\seq{}, \dots, \seq{\ty}, \dots, \seq{}})\qquad
\qquad  \llbracket \lambda x. M \rrbracket_{\vec{x}} ( \Delta, a) = \begin{cases} \llbracket M \rrbracket_{\vec{x} \oplus \seq{x}}(\Delta \oplus \seq{ \vec{a}'}, a') & \text{ if  } a = \iota(\vec{a}', a') \\ \emptyset & \text{ otherwise.}  \end{cases}  \end{gather*}
 \[ \llbracket MN \rrbracket_{\vec{x}}( \Delta, a ) = \int^{\vec{a} = \seq{a_{1},\dots, a_{k}}\in SD} \int^{\Gamma_{0}, \dots, \Gamma_{k} \in SD^{n}}  \llbracket M \rrbracket_{\vec{x}}(\Gamma_{0},\iota(\vec{a},a)) \times \prod\limits_{i=1}^{k} \llbracket N \rrbracket_{\vec{x}}(\Gamma_{i}, a_i) \times S D^{n}( \Delta, \bigotimes_{i = 0}^{k}\Gamma_{i}) \]
\caption{Denotation of $ \lambda $-terms.}
\hrulefill
\label{denL}
\end{figure*}
From now on we restrict ourselves to resource monad that have a symmetric tensor product\footnote{This is crucial since the Seely equivalence is needed in order to establish the cartesian closure.}. Thanks to Theorem \ref{Lresmon} and Proposition \ref{LpsA}, given a resource monad $S ,$ we obtain a (relative) pseudomonad $ \tilde{S} $ over distributors.
 We denote as $\SDist $ the Kleisli bicategory for this pseudomonad. We define the bicategory of $ S$-\emph{categorical symmetric sequences}, as $\SCat = \SDist^{op}  .$  It is useful to give an explicit definition of the relevant structure of  $ \SCat . $ When we write $ SA^o $ (resp. $SA^n $) we always mean $ (SA)^o  $(resp $(SA)^n $).
\begin{enumerate}
  \item $ \Ob{\SCat} = \Ob{\Cat} . $
  \item For $A, B \in \SDist $, we have $\SCat(A,B) = \SDist (B,A) = \Dist(B,SA) . $ 
  \item The identity is defined as \[  1_{A}(\vec{a},a) = SA(\vec{a},\seq{a}) .\]  
 \item For $ F : A \rightsquigarrow B $ and $ G : B \rightsquigarrow C $ $S $-categorical symmetric sequences, composition is given by considering $ F$ and $G $ as $S$-distributors:
\[ (G \circ F)(\vec{a}, c) = \int^{\vec{b} \in SB} G (\vec{b}, c) \times F^{\flat}(\vec{a}, \vec{b})        \]
where 
\[   F^{\flat}(\vec{a}, \vec{b}) = \int^{\tyl_1, \dots, \tyl_{\length{\vec b}}}      \prod_{i=1}^{\length{\vec{b}}} F(\tyl_i, b_i) \times SA(\bigoplus_{i=1}^{\length{\vec{b}}} \tyl_i, \tyl ) .\]
\item $ \SCat $ is cartesian. The cartesian product is the disjoint union $ A \with B = A \sqcup B $ and the projections are defined as follows:
\[ \pi_{i,2}(\vec{c},a) = S(A \sqcup B)(\vec{c}, \seq{\iota_{i} (a)})  .    \] 
The terminal object is the empty category.
\item The bicategory $ \SCat $ is cartesian closed, with exponential object $A \Rightarrow B = SA^o \times B . $
\end{enumerate}   

Indeed, exploiting the Seely equivalence we get the following chain of equivalences: \[    \SCat(A \with B, C)  =         \Dist(C, S(A \sqcup B)) =    \]\[         \CAT( S(A \sqcup B)^o \times C, \Set)  \simeq  \CAT(SA^o \times( SB^o  \times C), \Set)  =        \]
\[ \SDist( SB^o  \times C, A)     =             \SCat(A, SB^o  \times C)                            . \]

\section{ Models for pure $ \lambda $-calculus}\label{sem}

 We build a family of non-extensional bicategorical models for pure $\lambda$-calculus. These models will then be syntactically presented as appropriate categories of intersection types.

\begin{definition}
Let $A $ be a small category. We define by induction a family of small categories as follows:
\begin{gather*}
D_{0}= A \qquad 
D_{n+1} = ( S D^{o}_{n} \times D_n) \sqcup A
\end{gather*}

We define by induction on $ n \in \mathbb{N} $ a sequence of inclusions $ \iota_{n} : D_{n} \hookrightarrow D_{n+1}$: 
\begin{gather*}
\iota_{0} = \iota_A \qquad
\iota_{n+1} = (S(\iota_{n})^{o} \times \iota_n ) \sqcup 1_{A}
\end{gather*}
Then we set $ D_A =  \colim\limits_{ n \in \mathbb{N}  } D_{n}.$ 
\end{definition}
\begin{figure*}[t]	\scalebox{0.9}{\parbox{1.05\linewidth}{
\begin{align*} 
 \left(\begin{prooftree}
\hypo{f_1 : \vec{a}_1 \to \seq{}, \dots,  f_i : \tyl_i \to \seq{\ty}, \dots, f_{n} : \tyl_{n} \to \seq{} }
	\infer1[]{ x_1 :\tyl_1, \dots, x_i : \tyl_i, \dots x_n : \tyl_n \vdash x_i : \ty }
\end{prooftree} \right)\{ \eta \} \quad &= \quad \begin{prooftree}
\hypo{f_1 \circ g_1 : \vec{b}_1 \to \seq{}, \dots,  f_i \circ g_i : \vec{b}_i \to \seq{\ty}, \dots, f_{n} \circ g_n : \vec{b}_{n} \to \seq{} }
	\infer1[]{ x_1 :\vec{b}_1, \dots, x_i : \vec{b}_i, \dots x_n : \vec{b}_n \vdash x_i : \ty }
\end{prooftree}  \\ 
 \qquad \qquad \qquad  \left(\begin{prooftree}
\hypo{\pi}\ellipsis{}{ \Delta, x: \tyl \vdash \Tm : \ty }
	\infer1[]{ \Delta \vdash \lambda x. \Tm : \tyl \Rightarrow \ty }
\end{prooftree}\right)\{ \eta \} \quad &= \quad  \begin{prooftree} 
\hypo{\pi\{ \eta \oplus \seq{1}\}}\ellipsis{}{ \Delta', x: \tyl \vdash \Tm : \ty 
 }
	\infer1[]{ \Delta' \vdash \lambda x. \Tm : \tyl \Rightarrow \ty }
\end{prooftree} \\  
\left(\begin{prooftree}
\hypo{\pi_1}\ellipsis{}{\Gamma_{1} \vdash \Tm \colon \tyl \Rightarrow \ty }\hypo{\pi_i }
\ellipsis{}{\Gamma_{i} \vdash \Tm : \ty_i}
\delims{\left(}{\right)_{i =1}^{k}}
\hypo{ \theta : \Delta \to \bigotimes_{j=0}^{k} \Gamma_j }
	\infer3[]{ \Delta \vdash \appl{\Tm}{\Tmtwo} \colon \ty}\end{prooftree}\right)\{ \eta  \} \quad
		 &=  \quad
	\begin{prooftree}
\hypo{\pi_1}\ellipsis{}{\Gamma_{1} \vdash \Tm \colon \tyl \Rightarrow \ty }\hypo{\pi_i }
\ellipsis{}{\Gamma_{i} \vdash \Tm : \ty_i}
\delims{\left(}{\right)_{i =1}^{k}}\hypo{  \theta \circ \eta : \Delta' \to  \bigotimes_{j=0}^{k} \Gamma_j }
	\infer3[]{ \Delta' \vdash \appl{\Tm}{\Tmtwo} : \ty}
\end{prooftree} 
\end{align*} }}

Where $ \tyl = \seqdots{\ty}{1}{k} $ and $ \eta = \seqdots{g}{1}{n} : \Delta' \to \Delta . $
\caption{Right action on derivations.} 
\hrulefill
\label{la}
\end{figure*}
The category $D_A $ is the filtered colimit for the diagram $( D_{n} \hookrightarrow D_{n+1} )_{n \in \mathbb{N}} .$ 

This definition is actually a special case of the standard free-algebra construction for an (unpointed) endofunctor \cite{kelly:transc}. In our case the endofunctor is  $ S(-)^{o} \times (-) : \Cat \to \Cat $ and the free algebra is $ \seq{ D_A ,\iota : SD_{A}^o \times D_A \to D_A },$ where $ \iota $ is a canonical embedding. This determines a retraction $ {D_A \Rightarrow D_A} \lhd {D_A} $ in the bicategory $ \SCat. $ The reatraction pair is given by \begin{gather*}
\begin{aligned}
i :  (S (S D_{A}^{o} \times D_{A}))^o \times D_{A} \to \Set
\\ \seq{\vec{d}, a } \mapsto S D_{A}(S ( \iota)(\vec d), \seq{a})
\end{aligned}
\\ \\
 j : SD_{A}^o \times ( S D_{A}^{o} \times D_A ) \to \Set
\\ \qquad
\seq{\vec{a}', \seq{\vec{a},a}} \mapsto S D_{A} (\vec{a}', \seq{\iota(\vec{a},a)})
\end{gather*}
Hence, $ D_A $ is a (weak) reflexive object\footnote{Weak in this case means that the retraction condition is satisfied only up to canonical invertible 2-cell.}. If we set $ \seq{a_1, \dots, a_k} \Rightarrow a ::= \iota(\seq{a_1, \dots a_k},a), $ we can give a completely type-theoretic presentation of the category $D_A $ as in Figure \ref{IntD1}.

We fix a countable set of variables $ x, y, z, \dots \in   \mathcal{V} . $ The set of \emph{$ \lambda $-terms} is defined by induction  in the usual way:
\[    M, N \in \Lambda ::=       x  \mid \la{x} M \mid MN .\]
Terms are considered up to renaming of bound variables. As usual, we assume that application associates to the left. We denote the capture-free parallel substitution of variables as $\subst{M}{x_1, \dots, x_n}{N_1, \dots, N_n} .$ Given a term $M$, a list of terms $ \vec{N} = \seqdots{N}{1}{n}$ and a list of variables $ \vec{x} = \seqdots{x}{1}{m} $ we set $ M \vec{N} = MN_1 \dots N_n , \la{\vec{x}} M = \la{x_1} \dots \la{x_m} M.   $

The \emph{interpretation} of $ \lambda $-terms in the bicategory $\SCat $ is given by induction, following the standard categorical definition (\cite[Section 4.6]{ama:dom}). We fix a small category $A $ and a constant type $D $  such that $ D = D \Rightarrow D .\footnote{It is worth noting that we do not require for this equation to be semantically satisfied, \textit{i.e.} we consider non-extensional models.  } $
\begin{enumerate}
 \item On types: 
 \[ \sem{}{D} = D_A \qquad \sem{}{ \Gamma = D, \dots, D } =    \overbrace{ \sem{}{D} \with \dots \with \sem{} {D} }^\text{ n \text{ times } }    \] 
 \item On terms: 
 \[ \llbracket x_{1} : D, \dots, x_{n} : D \vdash x_{i} : D \rrbracket = \pi_{i,n} \] \[ \llbracket \Gamma \vdash \lambda x. M : D \rrbracket = i \circ \lambda(\llbracket \Gamma, x: D \vdash  M : D \rrbracket) \]
 \[\llbracket \Gamma \vdash PQ : D \rrbracket = ev_{D,D} \circ  \langle j \circ \llbracket \Gamma \vdash P : D\rrbracket,  \llbracket \Gamma \vdash Q : D \rrbracket \rangle . \]
  \end{enumerate} 
 Where $ \seq{i,j} $ is an appropriate retraction pair. Given $ \Gamma = x_1 : D, \dots, x_n :D $ we set $ supp(\Gamma) = \seq{ x_1, \dots, x_n } . $

\subsubsection{Denotations of Terms}\label{intdist}
The category $ D_A $ is a non-extensional model for pure $ \lambda $-calculus. We will denote, with a small abuse of language, $ SD_A $ as $SD $ and $D_A $ as $D .$   We now want to make explicit the idea that the semantics induced by this category is an intersection type system. In order to do so, we are going to define a parallel semantics, that we call the \emph{denotation} of a $ \lambda$-term. The intuition is that the denotation is the type-theoretic presentation, up to isomorphism, of the categorical semantics.

  We call \emph{ intersection type contexts}, or contexts for short, the objects of $ SD^n . $ Since $SD $ is monoidal,  the category $  SD^n  $ admits a tensor product, that we denote as $ \otimes, $ defined as follows: for $ \Gamma = \seqdots{\tyl}{1}{n} , \Delta = \seqdots{\vec{b}}{1}{n}$ we set $ \Gamma \otimes \Delta = \seq{\tyl_1 \oplus \vec{b}_1, \dots, \tyl_n \oplus \vec{b}_n} $. This tensor product inherits all the structure from $ \oplus , $ \textit{i.e.}, if $ \oplus  $ is symmetric (resp., semicartesian, relevant, cartesian) then also $ \otimes $ is so. 

We define the \emph{denotation} of a $ \lambda $-term by induction in Figure \nolinebreak\ref{denL}. We have that $ \sem{\vec{x}}{M} : D \todi SD^{n} . $ 

The denotation of an application is defined \textit{via} the Day convolution as follows. Consider the functor
\[        F : SD^{o} \times SD \times (SD ^n)^o \times D \to \Set                                     \]
\[     \seq{\tyl, \vec{b} = \seqdots{b}{1}{k}, \Delta, \ty} \mapsto        \]
\[  \left( {\sem{\vec{x}}{M} (-, \tyl \Rightarrow \ty)} \ \hat{\otimes} \  { \widehat{\bigotimes}_{i=1}^{k}   \sem{\vec{x}}{N}(-, b_i)  } \right) (\Delta)  \]
 the denotation of an application is then the following coend:
\[  { \sem{\vec{x}}{MN}(\Delta, \ty)} = \int^{\vec{a} \in SD} F(\vec{a}, \vec{a}, \Delta, a)  \]
 The action on morphism is given by the universal property of the coend construction.

The denotation of a term is isomorphic to its bicategorical interpretation \textit{via} the Seely equivalence:

\begin{theorem}\label{semden}
Let $ M \in \Lambda $, $ \vec{x} \supseteq FV(M)$ and $ \Gamma \vdash M : D $ such that $ supp(\Gamma) = \vec{x} $. We have a natural isomorphism \[ \llbracket M \rrbracket_{\vec{x}} \cong  {\overline{\mu}_1 \circ_{\Dist}  \llbracket \Gamma \vdash  M : D \rrbracket }. \]
\end{theorem}
\begin{proof}
By induction on the structure of $M $, \textit{via} lengthy but straightforward coend manipulations. 
\end{proof}

\subsection{The Denotation as an Intersection Type System} \label{denint}
\begin{figure*}[t] 	\scalebox{0.9}{\parbox{1.05\linewidth}{
\begin{align*} 
 [g : a \to b]\left(\begin{prooftree}
\hypo{f_1 : \vec{a}_1 \to \seq{}, \dots,  f_i = \seq{\alpha, f} : \tyl_i \to \seq{\ty}, \dots, f_{n} : \tyl_{n} \to \seq{} }
	\infer1[]{ x_1 :\tyl_1, \dots, x_i : \tyl_i, \dots x_n : \tyl_n \vdash x_i : \ty }
\end{prooftree} \right)  &= 
\begin{prooftree}
\hypo{f_1 : \vec{a}_1 \to \seq{}, \dots,  \seq{g} \circ f_i = \seq{\alpha, g \circ f} : \tyl_i \to \seq{b}, \dots, f_{n} : \tyl_{n} \to \seq{} }
	\infer1[]{ x_1 :\tyl_1, \dots, x_i : \tyl_i, \dots x_n : \tyl_n \vdash x_i : b}
\end{prooftree}  \\ 
 [\seq{\alpha, \vec{g}} \Rightarrow g : \vec{a} \Rightarrow a \to \vec{b} \Rightarrow b] \left(\begin{prooftree}
\hypo{\pi}\ellipsis{}{ \Delta, x: \tyl \vdash \Tm : \ty }
	\infer1[]{ \Delta \vdash \lambda x. \Tm : \tyl \Rightarrow \ty }
\end{prooftree} \right)  &=   \begin{prooftree} 
\hypo{([g] \pi) \{\seq{1, \seq{\alpha,\vec{g}} } \} }\ellipsis{}{ \Delta, x: \vec{b} \vdash \Tm : b }
	\infer1[]{ \Delta \vdash \lambda x. \Tm : \vec{b} \Rightarrow b }
\end{prooftree}   \\ 
[g : a \to b]\left(\begin{prooftree}
\hypo{\pi_0}\ellipsis{}{\Gamma_{0} \vdash \Tm \colon \tyl \Rightarrow \ty }\hypo{\pi_i }
\ellipsis{}{\Gamma_{i} \vdash \Tm : \ty_i}
\delims{\left(}{\right)_{i =1}^{k}}
\hypo{ \morpCone : \Delta \to \bigotimes_{0}^{k} \Gamma_j }
	\infer3[]{ \Delta \vdash \appl{\Tm}{\Tmtwo} \colon \ty}\end{prooftree}\right)
 &= 
	\begin{prooftree}
\hypo{[1 \Rightarrow g] \pi_0}\ellipsis{}{  \Gamma_{0} \vdash \Tm : \tyl \Rightarrow b  }\hypo{\pi_i }
\ellipsis{}{\Gamma_{i} \vdash \Tm : \ty_i}
\delims{\left(}{\right)_{i =1}^{k}}
\hypo{ \morpCone : \Delta \to  \bigotimes_{j= 0}^{k} \Gamma_j    }
	\infer3[]{ \Delta \vdash \appl{\Tm}{\Tmtwo}: b }
\end{prooftree}
\end{align*}}}

Where $ \tyl = \seqdots{\ty}{1}{k} . $
\caption{Left action on derivations.} 
\hrulefill
\label{ra}
\end{figure*}

We now give a type-theoretic description of the denotation of a $\lambda $-term. We define the intersection type system $E^{S}_A $, where types and morphisms live in the category $D_A$ (Figure \ref{IntD1}). Thanks to this type theoretic description, we can present the denotation's action on morphism as right and left actions on typing derivations: 
\[\begin{prooftree}
\hypo{ \pi }
\ellipsis{}{ \Delta &\vdash M : a} 
\end{prooftree} 
\qquad \rightsquigarrow \qquad
\begin{prooftree}
\hypo{
 ([f] \ \pi \ ) \{\morpCone\} }
\ellipsis{}{ \Delta' &\vdash M : a' } 
\end{prooftree}\]
with $ f : a \to a' $ and $ \morpCone : \Delta' \to \Delta . $ The actions are inductively defined in Figures \ref{la} and \ref{ra}. By an easy inspection of the definitions, we get $ (\pi \{ \eta \}) \{ \theta \} = \pi \{ \eta \circ \theta \} , ([f] \pi ) \{ \eta \} = [f]( \pi \{ \eta \} )  $ and $  [g] ([f] \pi) = [g \circ f] \pi .  $

We observe that in the variable rule of our system (Figure \ref{IntD2}) the morphisms $ f_j : \tyl_j \to \seq{} $ for $ j \neq i \in [n], $ are unique, by the structure of resource monads. In particular, if $S $ is \emph{irrelevant} (cartesian or semicartesian resource monad), $f_j $ is the the terminal morphism $ \top_{\vec{a}_j} : \vec{a}_j \to \seq{} .$ Otherwise (linear or relevant resource monad) $ f_j $ is the identity $ 1_{\seq{}} : \seq{} \to \seq{} .$

\subsubsection{Congruence on Typing Derivations}

The definition of denotation of an application $ MN $ depends on the notion of coend. In the $ \Set $ enriched setting this notion boils down to an appropriate quotient sum of sets. Hence, if we want to give a syntactic presentation of the denotation \textit{via} the intersection type system $E^{S}_A ,$ we shall need to translate the quotient in the setting of typing derivations. 

We set $ \tilde{\pi} $ as the equivalence class of $ \pi $ for the smallest congruence generated by the rules of Figure \ref{congr}. By an easy inspection of the definitions, we get that if $ \pi \sim \pi' $ then $ \pi \{ \theta \} \sim \pi \{ \theta \} $ and $ [g] \pi \sim [g] \pi' . $

  \begin{definition} Let $ \vec{x} \supseteq \fv{\Tm} $ and $ \length{\vec{x}} = n . $ We now define the \emph{S-intersection type distributor of \(M\)}, $ T_{D}(M)_{\vec{x}} : D \todi S D^n, $ as follows:
\begin{enumerate}
\item on objects \[  T_{D}(M)_{\vec{x}}(\Delta, a) = \left \{  \begin{prooftree}
\hypo{ \tilde{ \pi }  }
\ellipsis{}{ \Delta &\vdash M : a} \end{prooftree} \right\}  \]  
\item on morphisms \[ T_{D}(M)_{\vec{x}}(f, \morpCone) : T_{D}(M)_{\vec{x}}(\Delta, a) \to  T_{D}(M)_{\vec{x}}( \Delta', a')\]
\[         \tilde{\pi}   \mapsto \widetilde{  [f]  {\pi} \{ \morpCone \} } \]
\end{enumerate} 
\end{definition}

\begin{theorem}\label{ITD}
Let $ M \in \Lambda$. We have a natural isomorphism
\[    \sem{\vec{x}}{M} \cong T_{D}(M)_{\vec{x}} .
         \]
         \end{theorem}
\begin{proof}
By induction on the structure of $ M .$ The only non-trivial part is showing that in the application case, the distributor $  T_{D}(M)_{\vec{x}} $ can be described as a coend.
\end{proof}

\subsubsection{Typing Derivations under Reduction}\label{intred}

In this section we will prove that  $\sem{\vec{x}}{M}(\Delta, a) \cong \sem {\vec{x}} {N}(\Delta, a) $ when $ M \to_{\beta} N, $ refining the standard subject reduction and expansion for intersection types. Indeed, we recall that, by Theorem \ref{semden},  $\sem{\vec{x}}{M} \cong T_{D}(M)_{\vec{x}} $ 
hence, if we prove that $\sem{\vec{x}}{M}(\Delta, a) \cong \sem {\vec{x}} {N}(\Delta, a) $ when $ M \to_{\beta} N, $ in particular we have $T_{D}(M)_{\vec{x}} \cong T_{D}(N)_{\vec{x}}  .$ This means that we have a natural bijection between the set of equivalence classes of typing derivations with conclusion $ \Delta \vdash M : \ty $ and the set of equivalence classes of typing derivations with conclusion $ \Delta \vdash N : \ty , $ that is what we called a \emph{proof relevant} denotational semantics.

Let $ M, N \in \Lambda ,  (\fv{M} \setminus \{ x \}) \cup \fv{N} \subseteq \vec{x}  $ and $ x \notin \vec{x} . $  We set $ Sub^{M,x,N}_{\vec{x}}( \Delta, a) =  $ \[ \int^{\tyl \in SD} \int^{\Gamma_{j} \in S D^{n} } \sem{\vec{x} \oplus \seq{\vec{x}}}{M}(\Gamma_0 \oplus \seq{\tyl}, a) \times\] \[ \prod\limits_{i=1}^{l(\vec{a})} \sem{\vec{x}}{N}(\Gamma_{i}, a_i) \times S D^{n}(  \Delta, \Xi  )        \] where $ \Xi = \bigotimes_{j=0}^{\length{\tyl}} \Gamma_j . $ The intuition behind the former distributor is that its action on objects represent the structure that implicitly one considers in standard subject reduction and expansion lemmas\footnote{Reading the integral as an existential quantifier and the product as a conjunction, this analogy should be quite evident.}. We can now state the following (de)substitution lemma:

\begin{lemma}\label{sub}
Let $ M, N \in \Lambda ,  (\fv{M} \setminus \{ x \}) \cup \fv{N} \subseteq \vec{x}  $ and $ x \notin \vec{x} . $ We have a natural isomorphism

\[ \subis{M}{x}{N}{\vec{x}} : \sem{\vec{x}}{M[N/x]} \cong  Sub^{M,x,N}_{\vec{x}}.  \]

\end{lemma}

\begin{proof}
By induction on the structure of $ M $ \textit{via} lengthy coend manipulations. It is worth noting that, in the proof of the application case, the hypothesis about the symmetry of the tensor product over $ SD$ is crucial, as expected.
\end{proof}

\begin{theorem} \label{sound}
Let $ M , N \in \Lambda ,  \vec{x} \supseteq \fv{M} \cup \fv{N} $ and $ M \to_{\beta} N . $ We have  a natural isomorphism

 \[    \sem{\vec{x}}{M \to_{\beta} N} :   \sem{\vec{x}}{M} \cong \sem{\vec{x}}{N} .\]
 \end{theorem}
\begin{proof}
By induction on the reduction step $ M \to_{\beta} N . $ The base case is an immediate consequence of the former lemma.
\end{proof}

\section{Head-Normalization}\label{h-norma}

In this section we present a parametric head-normalization theorem for our systems, adapting the reducibility argument of \cite{carv:sem, kri:lam} to our setting. The construction of the argument is classical, but there is a technical improvement to be made in order to lift it to a category-theoretic perspective. 

We remark that our argument can be adapted also to characterize weak and strong normalization, as shown in \cite{ol:thesis}.
Given a $  \lambda$-term $ M ,$ we denote as $ H(M) $ its \emph{head-reduct}\footnote{In the definition we use a well-known characterization of $ \lambda$-terms, see \cite{kri:lam}.}:
\[      H(M) = \begin{cases}   M  & \text{ if } M =    \la{\vec{x}} x \vec{N}\\  \la{\vec{x}} \subst{M}{x}{N} \vec{N} & \text{ if } M =  \la{\vec{x}} (\la{x} M) N \vec{N}    .    \end{cases}                  \]
 If $ H(M) = M $ we say that $ M $ is a \emph{head-normal form.}

\begin{figure*}	\scalebox{0.9}{\parbox{1.05\linewidth}{
\begin{align*}\label{eq1}
 \begin{prooftree}
\hypo{\pi_0}\ellipsis{}{ \Gamma_{0} \vdash \Tm : \vec{b} \Rightarrow \ty }\hypo{[f_{i}]\pi_{\alpha(i)}}
\ellipsis{}{\Gamma_{\alpha(i)} \vdash \Tm : b_i}
\delims{\left(}{\right)_{i =1}^{k'}}
\hypo{ (1 \otimes \alpha^{\star}) \circ \morpCone : \Delta \to \Gamma_0 \otimes \bigotimes_{i = 1}^{k'} \Gamma_{\alpha(i)} }
	\infer3[]{ \Delta \vdash \appl{\Tm}{\Tmtwo} : \ty }
\end{prooftree}  &\sim   \begin{prooftree}
\hypo{[ \seq{\alpha, \vec{f}} \Rightarrow 1]  \pi_0 }\ellipsis{}{ \Gamma_{0} \vdash \Tm : \tyl \Rightarrow \ty} \hypo{\pi_i }
\ellipsis{}{\Gamma_{i} \vdash N : \ty_i}
\delims{\left(}{\right)_{i =1}^{k}}
\hypo{  \morpCone : \Delta \to \bigotimes_{j=0}^{k} \Gamma_{j} }
	\infer3[]{ \Delta \vdash \appl{\Tm}{\Tmtwo} : \ty }
\end{prooftree} \\
 \begin{prooftree}
\hypo{\pi_0 \{ \theta_0 \} } \ellipsis{}{   \Gamma_{0} \vdash \Tm : \tyl \Rightarrow \ty}\hypo{\pi_i \{\theta_i \} }
\ellipsis{}{\Gamma_{i} \vdash N : \ty_i}
\delims{\left(}{\right)_{i =1}^{k}}
\hypo{ \morpCone : \Delta \to \bigotimes_{j =0}^{k} \Gamma_j }
	\infer3[]{ \Delta \vdash \appl{\Tm}{\Tmtwo} : \ty }
\end{prooftree}  &\sim   \begin{prooftree}
\hypo{\pi_0}\ellipsis{}{  \Gamma'_{0} \vdash \Tm : \tyl \Rightarrow \ty}\hypo{\pi_i }
\ellipsis{}{\Gamma'_{i} \vdash N : \ty_i}
\delims{\left(}{\right)_{i =1}^{k}}
\hypo{   (\bigotimes_{j=0}^{k} \theta_j) \circ  \morpCone}
	\infer3[]{ \Delta \vdash \appl{\Tm}{\Tmtwo} : \ty}
\end{prooftree}\end{align*} }}

Where $   \seq{\alpha, f_1, \dots, f_{k'}} : \tyl = \seqdots{\ty}{1}{k}  \to \vec{b} = \seqdots{b}{1}{k'}  $ and $ \theta_i : \Gamma_i \to   \Gamma'_{i}.$
\caption{Congruence on typing derivations.} 
\hrulefill
\label{congr}
\end{figure*} 
\begin{lemma} 
Let $ M \in \Lambda $  be a head-normal form. Then  $ \sem{\vec{x}}{M} \neq \emptyset_{D_A, SD_{A}^{\length{\vec{x}}}} . $  
\end{lemma}
\begin{proof}

We have that $ M = \lambda x_1 \dots \lambda x_m . x Q_1 \cdots Q_n . $ We prove it for $ x Q_1 \cdots Q_n ,$ choosing as list of variables $ \vec{y}  = \vec{x} \oplus \seqdots{x}{1}{m} = \seqdots{y}{1}{k}     $ where $ k = m + \length{\vec{x}}  $, the extension being immediate.  

Let $ b =   \seq{} \Rightarrow \dots \Rightarrow \seq{} \Rightarrow a .$   It is enough to take the following typing derivation $ \pi = $

\[\begin{prooftree}
\hypo{ 1_{\seq{}}, \dots ,  1_{\seq{b}} , \dots,  1_{\seq{}}               }
\infer1{y_1 : \seq{}, \dots, x: \seq{b}, \dots, y_k  : 1_{\seq{}} \vdash x : \seq{} \Rightarrow \dots \Rightarrow \seq{} \Rightarrow a }
\infer{1}{y_1 : \seq{}, \dots,  x: \seq{b}, \dots, y_k : \seq{} \vdash x Q_1 \cdots Q_n : a}
\end{prooftree}\]

Then $ T_{D}(M)_{\vec{y}}(\seq{\seq{}, \dots, \seq{b}, \dots, \seq{}}, \ty  ) $ is non-empty for all types $ \ty \in D . $ Then we apply Theorem \ref{semden} and conclude. 


\end{proof}

\begin{corollary}\label{hnormal}
Let $ M \in \Lambda . $ If $ M $ is head-normalizable then $\sem{\vec{x}}{M}  \neq \emptyset_{D_A, SD_{A}^{\length{\vec{x}}}}  . $
\end{corollary}
\begin{proof}
Corollary of the former lemma and Theorem \ref{sound}.
\end{proof}

We are now ready to present our reducibility argument. For a set $\mathcal{X} \subseteq \Lambda  $ we say that $ \mathcal{X} $ is \emph{saturated} if $ \subst{M}{x}{N} N_1 \dots N_n \in \mathcal{X}   $ implies $ ((\lambda x . M)N)N_1 \dots N_n \in \mathcal{X} . $ Given $ \mathcal{X}_1, \mathcal{X}_2 \subseteq \Lambda ,$ we write $ \mathcal{X}_1 \Rightarrow \mathcal{X}_2 = \{ M \in \Lambda  \mid \text{ for all } N \in \mathcal{X}_1, MN \in \mathcal{X}_2    \} .$ 

Given a small category  $ A $ an \emph{interpretation} is a functor $ I : A \to ((\wp \Lambda)^{\ast}, \subseteq) ,$ where $ (\wp \Lambda)^{\ast} = \{ \mathcal{X} \subseteq \Lambda \mid \mathcal{X}  \text{ is saturated }   \} . $ Given $ \delta \in D_{A} \sqcup S D_A $ we define the \emph{set of realizers} of $ \delta $ by induction as follows:
\[   \sem{I}{o} = I(o) \qquad \sem{I}{\seq{}} = \Lambda \qquad \sem{I}{\seq{a_0, \dots, a_k}} = \bigcap                        _{i= 0}^{k} \sem{I}{a_i}  \] \[ \sem{I}{\vec{a} \Rightarrow \ty} = \sem{I}{\tyl} \Rightarrow \sem{}{\ty}  \]

By construction we have that $ \sem{I}{\delta} $ is saturated. 
 
\begin{lemma}\label{functR}
Let $ \delta, \delta' \in D_{A} \sqcup S D_A . $ If $ f : \delta \to \delta' $ then $ \sem{I}{\delta} \subseteq \sem{I}{\delta'} . $
\end{lemma} 
\begin{proof}
By induction on the structure of $ \delta' .$ 
\end{proof} 

\begin{lemma}\label{adeq}
Let $ M , N_1, \dots, N_n \in \Lambda  $ and $ I $ be an interpretation. If $ x_1 : \tyl_1, \dots, x_n : \tyl_n \vdash M : \ty $ and $  N_i \in \sem{I}{\vec{a}_i} $ then $   \subst{M}{x_1, \dots, x_n}{N_1, \dots, N_n} \in \sem{I}{\ty} .     $
\end{lemma}
\begin{proof}
By induction on the structure of $M , $ applying Lemma \ref{functR}. 
\end{proof}

We define  $ \mathcal{HN} = \{ M \in \Lambda \mid  \text{ The head-reduction of $M$ ends} \}$  and $ \mathcal{HN}_{0} = \{ xN_1 \dots N_n \mid N_i \in \Lambda  \} . $ We remark that $ \mathcal{V} \subseteq \mathcal{HN}_{0} . $

\begin{lemma}
 $\mathcal{HN}$ is saturated. 
\end{lemma}

 We set $ I_{HN} : A \to ((\wp \Lambda)^{\ast}, \subseteq ) $ to be the functor such that for all $ a \in A , $ $ I_{HN}(a) = \mathcal{HN} ,$ the action on morphisms being the trivial one. We define in the same way $ I_{N} $ and $ I_{SN} . $

\begin{lemma}\label{inclR}  For all $ a \in D_A $ we have that $ \mathcal{HN}_{0} \subseteq \sem{I_{HN}}{a} \subseteq \mathcal{ HN } . $

\end{lemma}

\begin{lemma}\label{nteo} Let $ M \in \Lambda . $  If $ M $ is typable in the system $ E^{S}_A $ then the head-reduction of $ M $ ends.\end{lemma}
\begin{proof}
Direct consequence of Lemma \ref{inclR} and Lemma \ref{adeq}. 
\end{proof}

\begin{theorem} Let $ M \in \Lambda. $ The following statements are equivalent.
\begin{enumerate}
\item  $  \sem{\vec{x}}{M} \neq \emptyset_{D, S D^{\length{\vec{x}}}} .  $
\item The head-reduction of $ M $ ends.
\item $ M $ is head-normalizable.
\end{enumerate} 
\end{theorem}\begin{proof} $ (1) \Rightarrow (2)$ Corollary of Theorems \ref{ITD} and Lemma \ref{nteo}. $ (2) \Rightarrow (3) $ immediate by definition. $ (3) \Rightarrow (1) $  by Theorem \ref{semden} and Corollary \ref{hnormal}.
\end{proof}

\section{Worked Out Examples}\label{app}

We present two concrete constructions of the distributor-induced denotational semantics that we introduced in the previous sections. We choose the examples of the linear resource monad (symmetric monoidal strict completion) and of the cartesian one (cartesian strict completion). Those two examples are particularly relevant since they correspond to the categorification of the two best known intersection type systems: the linear logic induced Gardner-De Carvalho System $ \mathcal{R} $ \cite{gard:int, carv:sem} and the original Coppo-Dezani System $ D \Omega $ \cite{dez:int}. The first one is non-idempotent, the second one is idempotent. In our setting, the idempotency issue is replaced by an operational one: which operations do we allow on intersections? 
 \subsection{Example 1: Linear Resources}
 The content of this subsection corresponds to a Call-by-Name version of \cite{ol:bang}. We present the non-idempotent intersection type system $ \mathcal{R} $. That system has a categorical counterpart in the linear logic induced relational model for pure $ \lambda $-calculus \cite{carv:sem}. The intersection type is given by multisets. In our case, we achieve a non-idempotent and commutative (up to isos) intersection type system applying our construction in the special case where the resource monad $S$ is the $ 2$-monad for symmetric strict monoidal categories. The corresponding intersection type system is system $ R_A $ in Figure \ref{intCart}.

In the linear case, we can prove the head-normalization theorem in a combinatorial way. We set $ \sem{\vec{x}}{M}^{R_A} $ to be the denotation of $ M $ in the case where $ S$ is the linear resource monad. We define the \emph{size} $ \size{\pi} $ of a typing derivation $\pi $ as the number of application rules that appear in it.

By an easy inspection of the definitions, we have that the size is stable under actions and under congruence: if $ \pi, \pi' \in R_A $ and $ \pi \sim \pi' $ then $ \size{\pi} = \size{\pi'} . $

Let $ \rho_{\Delta, \ty} : \sem{\vec{x}}{\Tm}^{R_A}(\Delta, \ty) \cong T_{D}(M)_{\vec{x}}^{R_A}(\Delta, \ty) $ be the isomorphism given by Theorem \ref{denint}. Then for $ \alpha \in  \sem{\vec{x}}{\Tm}(\Delta, \ty)$ we set $ \size{\alpha} = \size{\rho_{\Delta, \ty}(\alpha)} . $ Given $ \vec{\alpha} = \seqdots{\alpha}{1}{k} $ with $ \alpha_i \in \sem{\vec{x}}{M},  $ we set $ \size{\vec{\alpha}} = \sum_{i=1}^k \size{\alpha_i} . $

\begin{lemma}\label{fsub}
Let $ M, N $ be two $\lambda $-terms,  $ \vec{x} \supseteq {\fv{M} \cup \fv{N}   } $ with $ x \notin \vec{x} $ and
\[ \subis{M}{x}{N}{\vec{x}}_{\Delta, \ty}   : Sub_{\vec{x}}^{M,x,N}(\Delta, \ty) \cong  \sem{\vec{x}}{\subst{M}{x}{N}}^{R_A}(\Delta, \ty)   \]
be the natural isomorphism given by Theorem \ref{semden}. For all $ \alpha = \widetilde{\seq{\pi, \vec{\psi}, \eta}} \in Sub_{\vec{x}}^{M,x,N}(\Delta, \ty) , $ we have
\[   \size{\subis{M}{x}{N}{\vec{x}}_{\Delta, \ty} (\alpha) } = \size{\pi} + \size{\vec{\psi}} .   \]
\end{lemma}

\begin{theorem}
Let $ M, N  \in \Lambda  .$  We have a natural isomorphism\[ \varphi_{\Delta,a} : \sem{\vec{x}}{M}^{R_A}(\Delta,a) \cong \sem{\vec{x}}{H(M)}^{R_A}(\Delta, \ty)   \] such that for $ \alpha \in \sem{\vec{x}}{M}^{R_A}(\Delta,a) , \size{\varphi_{\Delta,a}(\alpha)} \lneq \size{\alpha} . $
\end{theorem}
\begin{proof}
Direct corollary of the former lemma.
\end{proof}

\begin{theorem}
Let $ M \in \Lambda $. If $ \sem {\vec{x}} {M}^{R_A}  \neq \emptyset_{D, SD^{\length{\vec{x}}}} $ the head reduction of $ M $ ends. \end{theorem}
\begin{proof}
We have that, for $ \varphi : \sem{\vec{x}}{M}^{R_A} \cong \sem{\vec{x}}{H(M)}^{R_A} . $ If $ \sem {\vec{x}} {M}^{R_A}  \neq \emptyset_{D, SD^{\length{\vec{x}}}} $  then $ \sem {\vec{x}} {H(\Tm)}^{R_A}  \neq \emptyset_{D, SD^{\length{\vec{x}}}}$. We consider $ \alpha \in \sem {\vec{x}} {M}^{R_A} (\Delta,a) $ for some $ \seq{\Delta, \ty} \in SD^{\length{\vec{x}}} \times D . $ Then, by the former theorem, $ \size{\varphi_{\Delta, \ty}} < \size{\alpha} . $ Then we can apply the IH 	and conclude.
\end{proof}

\begin{figure*}[t]

\begin{align*}
 \begin{prooftree}
\hypo{ f : a' \to a }
	\infer1[]{ x_1 :\seq{}, \dots, x_i : \seq{a'}, \dots, x_n :\seq{} \vdash x_i : a }
\end{prooftree} \qquad & \qquad \begin{prooftree}
\hypo{ \top_{\tyl_1} : \vec{a}_{1} \to \seq{}, \dots, f_i : \vec{a} \to \seq{a} , \dots, \top_{\tyl_n} :  \vec{a}_{n} \to \seq{}}
	\infer1[]{ x_1 :\vec{a}_1, \dots, x_i : \vec{a}, \dots x_n :\vec{a}_n \vdash x_i : a }
\end{prooftree} \\
    \begin{prooftree}
\hypo{ \Delta, x: \vec{a} \vdash M : a }
	\infer1[]{ \Delta \vdash \lambda x. M : \vec{a} \Rightarrow a }
\end{prooftree} \qquad & \qquad \begin{prooftree}
\hypo{ \Delta, x: \vec{a} \vdash M : a }
	\infer1[]{ \Delta \vdash \lambda x. M : \vec{a} \Rightarrow a }
\end{prooftree}  \\
 \begin{prooftree}
\hypo{ \Gamma_0 \vdash M : \vec{a} \Rightarrow a }\hypo{( \Gamma_i \vdash N : a_{i})_{i \in [k]} } \hypo{\eta : \Delta \to \bigotimes_{j = 0}^{k} \Gamma_j }
\infer3[]{ \Delta \vdash MN : a }
\end{prooftree} \qquad & \qquad
\begin{prooftree}
\hypo{ \Delta \vdash M : \vec{a} \Rightarrow a }\hypo{(\Delta \vdash N : a_{i})_{i \in [k]} }
\infer2[]{ \Delta \vdash MN : a }
\end{prooftree}
\end{align*}

Where $ \vec{a} = \seqdots{\ty}{1}{k} . $ 
\caption{Intersection type systems $R_A $ and $ C_A $.} 
\hrulefill
\label{intCart}
\end{figure*}

\begin{figure*}
\begin{gather*}
 \begin{prooftree}
\hypo{\pi_0}\ellipsis{}{ \Delta \vdash \Tm : \vec{b} \Rightarrow \ty }\hypo{[f_{i}]\pi_{\alpha(i)}}
\ellipsis{}{\Delta \vdash \Tm : b_i}
\delims{\left(}{\right)_{i =1}^{k'}}
	\infer2[]{ \Delta \vdash \appl{\Tm}{\Tmtwo} : \ty }
\end{prooftree} \quad \sim  \begin{prooftree}
\hypo{[ \seq{\alpha, \vec{f}} \Rightarrow 1]  \pi_0 }\ellipsis{}{ \Delta \vdash \Tm : \tyl \Rightarrow \ty} \hypo{\pi_i }
\ellipsis{}{\Delta \vdash N : \ty_i}
\delims{\left(}{\right)_{i =1}^{k}}
	\infer2[]{ \Delta \vdash \appl{\Tm}{\Tmtwo} : \ty }
\end{prooftree} \end{gather*} 
Where $   \seq{\alpha, f_1, \dots, f_{k'}} : \tyl = \seqdots{\ty}{1}{k}  \to \vec{b} = \seqdots{b}{1}{k'}  .$ \caption{Congruence on cartesian typing derivations.}\hrulefill
\label{cartcong}
\end{figure*}

\begin{example}
We provide a simple example of reduction of typing derivations to ease the understanding of the congruence's role in establishing the natural isomorphisms. Consider $ \Tm = (\la{x} x) y . $ We type it with the following typing derivations:
\[  \pi_1 =  \begin{prooftree}   \hypo{h \circ f : a \to b}\infer1{x : \seq{a} \vdash x : b}\infer1{ \vdash \la{x} x : \seq{a} \Rightarrow b }     \hypo{g : c \to a}\infer1{y : \seq{c} \vdash y : a}   \hypo{1} \infer3{ y : \seq{c} \vdash (\la{x} x) y : b }     \end{prooftree}      \] \[ \pi_2 =  \begin{prooftree}   \hypo{ h \circ f' : d \to b}\infer1{x : \seq{a} \vdash x : b}\infer1{ \vdash \la{x} x : \seq{d} \Rightarrow b }     \hypo{g' : c \to d}\infer1{y : \seq{c} \vdash y : d}   \hypo{1} \infer3{ y : \seq{c} \vdash (\la{x} x) y : b }     \end{prooftree}      \]
suppose that $ f \circ g =  f' \circ g'  $ and $ h : b \to b, f : a \to b, f' : d \to b . $ We have that $ \pi_1 \sim \pi_2 .$ Indeed, by the first rule of Figure \ref{congr}:
\[ \pi_1 \sim    \begin{prooftree}   \hypo{h : b \to b}\infer1{x : \seq{b} \vdash x : b}\infer1{ \vdash \la{x} x : \seq{b} \Rightarrow b }     \hypo{ f \circ g : c \to b}\infer1{y : \seq{c} \vdash y : b}    \hypo{1} \infer3{ y : \seq{c} \vdash (\la{x} x) y : b }     \end{prooftree}    \] \[    \pi_2 \sim    \begin{prooftree}   \hypo{h : b \to b}\infer1{x : \seq{b} \vdash x : b}\infer1{ \vdash \la{x} x : \seq{b} \Rightarrow b }     \hypo{ f' \circ g' : c \to b}\infer1{y : \seq{c} \vdash y : b}  \hypo{1} \infer3{ y : \seq{c} \vdash (\la{x} x) y : b }     \end{prooftree}    \] 
and by the hypothesis that $ f \circ g =  f' \circ g'$ we can conclude by transitivity. In particular, this means that the quotient identify all couple of morphisms leading to the same composition. Now, we have that $ \Tm \to y . $ Consider the following typing derivation of $ y : $
\[ \pi_3 = \begin{prooftree}   \hypo{h \circ (f \circ g) : c \to b}\infer1{y : \seq{c} \vdash y : b}     \end{prooftree}     \]
By an easy inspection of the definitions we have that for $  \varphi_{\seq{c}, b} : \sem{\seq{y}}{\Tm}(\seq{c}, b) \cong  \sem{\seq{y}}{y}(\seq{c}, b)$, $\varphi_{\seq{c}, b}(\tilde{\pi_1}) =  \pi_3 ,$ where we keep implicit the isomorphism given by Theorem \ref{intdist}. There is then a nice correspondence between \emph{substitution} on the term side and \emph{composition} on the morphism side, that validates the basic intuition of categorical semantics\footnote{The natural isomorphism $\varphi_{\seq{c}, b} : \sem{\seq{y}}{\Tm}(\seq{c}, b) \cong  \sem{\seq{y}}{y}(\seq{c}, b) $ is a particular instance of the Yoneda Lemma for coends \cite{fosco:coend}.  }.  
\end{example}

\subsection{Example 2: Cartesian Resources} \label{cart}

We now focus on the type theoretic semantics induced by the cartesian resource monad. In this framework, a resource can be copied and deleted at wish. 

When $SA $ is cartesian, the Day convolution on $ PSA $ is isomorphic to the cartesian product. Hence, we have the following natural isomorphism\footnote{Simply observing that the tensor product over $SD^{n} $ is cartesian and then exploiting the natural bijection given by the adjunction.}
\[    G \circ F (\vec{a},c) \cong \int^{\seqdots{b}{1}{k} \in SD} G(\vec{b}, c) \times \prod_{i \in [n] } F(\vec{a}, b_i)     \]
For $ F,G : D \rightsquigarrow D  $ in $\SCat. $ This Kleisli bicategory is known as the bicategory of \emph{cartesian distributors} \cite{fiore:cart}. By straightforward coend manipulations, we derive the type system $ C_A$ described in Figure  \ref{intCart}. Actions on typing derivations are defined in the straightforward way. The equivalence on typing derivation in this case is generated only by the rule of Figure \ref{cartcong}, since now the coend on contexts disappeared. It is worth noting that the cartesian category $SD_A $ admits all the basic axioms imposed on the preorder over idempotent intersection types \cite{dez:intl}. This means that our construction generalizes the standard subtyping relation, as expected. However, the two conditions 
\[ \pi_{i, 2} : \vec{a}_1 \oplus \vec{a}_2 \to \vec{a}_i         \qquad  c_{\vec{a}} : \vec{a} \to \vec{a} \oplus \vec{a}   \]
do not determine an idempotency $ \vec{a} \oplus \vec{a} \cong \vec{a} .$ In our categorified setting, idempotency is replaced by the possibility to perform two operations on resources: \emph{copying} and \emph{deleting}.   

\begin{example}
We provide some example of typing derivations in system $ C_A  ,$ giving also some intuition for what concerns the congruence on typing derivations. 
\begin{enumerate}
\item Let us type the term $ M = (\la{x} (x x) x) . $ Let $ b =  \seq{a} \Rightarrow \seq{a} \Rightarrow \ty.$ Consider the following typing derivation $ \pi : $

\scalebox{0.80}{\parbox{1.05\linewidth}{\[  \begin{prooftree}
\hypo{\pi_1 : \seq{b, \ty} \to \seq{b}}\infer1{x : \seq{b, \ty} \vdash x : \seq{a} \Rightarrow \seq{a} \Rightarrow \ty } \hypo{\pi_2 : \seq{b, \ty} \to \seq{\ty}}\infer1{x : \seq{b, \ty} \vdash x : \ty}  \infer2{x : \seq{b, \ty} \vdash xx : \seq{\ty} \Rightarrow \ty} \hypo{\pi_2 : \seq{b, \ty} \to \seq{\ty}}\infer1{x : \seq{b, \ty} \vdash x : \ty }\infer2{x : \seq{b, \ty} \vdash (xx)x : \ty} \infer1{\vdash \la{x} (xx)x : \seq{b, \ty} \Rightarrow \ty}
\end{prooftree}   \]}}
\end{enumerate}
\item Now consider the following typing derivation $ \rho : $

\[\begin{prooftree}
\hypo{\pi}\ellipsis{}{\vdash \la{x} (xx)x : \seq{b, \ty} \Rightarrow \ty} \hypo{\vdash N : b} \hypo{\vdash N :a}\infer3{\vdash (\la{x} (xx)x) N : \ty}
\end{prooftree} \]
and $ \pi'  :$

\scalebox{0.75}{\parbox{1.05\linewidth}{\[  \begin{prooftree}
\hypo{\pi_1 : \seq{b, \ty, \ty} \to \seq{b}}\infer1{x : \seq{b, \ty, \ty} \vdash x : \seq{a} \Rightarrow \seq{a} \Rightarrow \ty } \hypo{\pi_2 : \seq{b, \ty, \ty} \to \seq{\ty}}\infer1{x : \seq{b, \ty, \ty} \vdash x : \ty}  \infer2{x : \seq{b, \ty, \ty} \vdash xx : \seq{\ty} \Rightarrow \ty} \hypo{\pi_3 : \seq{b, \ty, \ty} \to \seq{\ty}}\infer1{x : \seq{b, \ty, \ty} \vdash x : \ty }\infer2{x : \seq{b, \ty, \ty} \vdash (xx)x : \ty} \infer1{\vdash \la{x} (xx)x : \seq{b, \ty, \ty} \Rightarrow \ty}
\end{prooftree}   \]}}

We have that $ \pi = [c^{\ast} \Rightarrow 1] \pi' $ where $ c^{\ast} = 1_{\seq{b}} \oplus c_{\seq{a}} . $ If we consider then the following derivation $ \rho' : $
\[\begin{prooftree}
\hypo{\pi'}\ellipsis{}{\vdash \la{x} (xx)x : \seq{b, \ty, \ty} \Rightarrow \ty} \hypo{\vdash N : b} \hypo{\vdash N :a} \hypo{\vdash N :a} \infer4{\vdash (\la{x} (xx)x) N : \ty}
\end{prooftree} \]
We have that $ \rho \sim \rho' $ by the rule of congruence of Figure \ref{cartcong}.
\end{example}

\section{Conclusions}

\subsubsection{Results}

Bringing together several independent results and perspectives, we gave a consistent argument in favour of considering the bicategory of distributors as an appropriate framework for a general theory of intersection types. We defined a family of Kleisli bicategories of distributors, parametric over a resource monad. We gave a sufficient condition for these Kleisli bicategories to be cartesian closed. We then defined non-extensional models for pure $\lambda $-calculus. We showed how each resource monad is equivalent to a particular intersection type construction. Each model that we presented can be seen as an appropriate \emph{category of types}. From this category of types we defined an intersection type system and, consequently, a \emph{proof relevant} denotational semantics. We then proved that these semantics are coherent with respect to solvability.

\subsubsection{Perspectives}

The flexibility of our approach opens a considerable amount of possible future investigations. From an abstract standpoint, it is tempting to go even a bit further in the direction of \cite{mazza:pol} and identify our construction of intersection type distributors with an interpretation morphism between the symmetric $ 2$-operad of $ \lambda $-terms and the bicategory $\Dist .$ This identification makes intuitively sense because of the strict connection between Kleisli bicategories of distributors and multicategories \cite{fiore:rel}. We leave all these speculations to future work. 

We believe that an appropriate presentation of the results of \cite[Chapter 4]{ol:thesis} about the relationship between intersection type distributors and $ \lambda $-terms \emph{rigid approximants} \cite{ol:group, tao:gen} would be of great interest. In particular, this would make explicit the relationship between our semantic approach and the one of \cite{tao:gen,tao:pdist}. 

The study of the relationship between our distributors induced semantics and Melliès \emph{template games} semantics \cite{mell:diff} would shed some light on the connection between intersection types, term approximants and game semantics. The natural starting point for this investigation would be a particular special case of the general construction of \cite{mell:diff}: the bicategory of spans over groupoids. What happens in that framework with the exponential modality is particularly interesting and could give a \emph{homotopic} flavour to our semantic perspective.

 Another possibility is the investigation of extensional collapse, in the sense of \cite{er:collapse}. We believe that connecting the approaches of the present work and of \cite{zeinab:prof} we could reach an understating of the semantic link between non-idempotent intersection types and idempotent ones in the bicategorical setting of distributors. 
 
 An extension  of our approach to probabilistic computation, algebraic $ \lambda $-calculus \cite{vaux:alg} would also be of great interest. This would be somehow related to \cite{tao:pdist}, \cite{manzo:nd} and \cite{dlago:pint}. 

Finally, another interesting question arises in the context of Multiplicative Exponential Linear Logic (MELL). Since the notion of experiment \cite{tort:exp} can be thought as the proof-net version of typing derivations, a possible extension of this work to that setting could give relevant information about the experiments reduction \cite{carv:weak}.


\printbibliography

\end{document}